\documentclass{article}
\usepackage[utf8]{inputenc}
\usepackage[T1]{fontenc}
\usepackage{amsmath, amssymb, amsfonts, amsthm, mathtools}
\usepackage{geometry}
\usepackage{hyperref}
\usepackage{authblk}

\newtheorem{theorem}{Theorem}
\newtheorem{lemma}{Lemma}
\newtheorem{proposition}{Proposition}

\theoremstyle{definition}
\newtheorem{definition}{Definition}
\theoremstyle{remark}
\newtheorem{remark}{Remark}

\DeclareMathOperator{\E}{\mathbb{E}}
\newcommand{\R}{\mathbb{R}}

\title{Temporal Matrix Scale Invariance and the Classification of Tipping Points}
\author[1,2]{Alejandro Frank\thanks{Member of El Colegio Nacional.}}
\author[1,3]{Laurence A. Jacobs\thanks{Correspondence: laurence.jacobs@uzh.ch}}
\affil[1]{Centro de Ciencias de la Complejidad, Universidad Nacional Autónoma de México, Mexico City, Mexico}
\affil[2]{Instituto de Ciencias Nucleares, Universidad Nacional Autónoma de México, Mexico City, Mexico}
\affil[3]{Center for Molecular Cardiology, University of Zürich, Switzerland}
\date{May 4, 2026}

\begin{document}

\maketitle

\begin{abstract}
We introduce temporal matrix scale invariance (tMSI), a mathematical structure for the two-time correlation kernel of a multivariate observable. A kernel $C(t, t')$ satisfies tMSI of order $\alpha$ if $C(kt, kt') = k^{-\alpha}C(t, t')$ for all $k > 0$; this condition holds at or near a tipping point, where the divergence of the coherence time produces temporal scale freedom. By a kernel factorization theorem, every tMSI kernel separates into a power-law envelope $(tt')^{-\alpha/2}$ and a shape function $F(t/t')$ diagonalized by the Mellin transform. This structure reveals a decoupling of two independent exponents: the dynamical exponent $\alpha$, carried by the envelope, and the spectral relaxation exponent $\beta$, determined by the eigenvalue decay of the finite-dimensional truncation. Their equality $\alpha = \beta$ characterizes a simple critical point; their inequality $\alpha \neq \beta$ is the signature of temporal multicriticality.

The main result is a complete classification of tipping points. The Landau quartic coefficient $a_4$ of the order parameter is given exactly by
\[
a_4 = p^2 + q^2 - 2\lambda p q - g^2_{\alpha\alpha\beta} \Gamma(\sigma_\alpha, \sigma_\beta),
\]
where $\lambda = 2\sqrt{\sigma_\alpha\sigma_\beta}/(\sigma_\alpha + \sigma_\beta) \in (0, 1]$, $g_{\alpha\alpha\beta}$ is the three-point structure constant, and $\Gamma > 0$ is in explicit closed form. The tipping point is continuous when $a_4 > 0$, tricritical when $a_4 = 0$, and discontinuous when $a_4 < 0$. The simple critical point $\alpha = \beta$ is maximally fragile: any nonzero operator mixing drives $a_4 < 0$, so the synchronized state generically sits at the edge of catastrophe.

The framework yields a matrix-valued early warning diagnostic, computable from a multivariate time series without knowledge of the underlying equations, that classifies an approaching tipping point as recoverable or catastrophic. Applications to epileptic seizure onset and acute myocardial infarction are discussed.
\end{abstract}


\section{Introduction}
The concept of temporal matrix scale invariance was introduced by Frank [1] as an empirical framework for detecting dynamical criticality in complex systems through the temporal covariance structure of multivariate data. The present paper provides the rigorous mathematical foundations of that framework.

The abrupt, often irreversible transitions that characterize tipping points—epileptic seizure onset, cardiac infarction, ecological collapse—present a fundamental diagnostic challenge: the system appears stable until it does not, and by the time the transition is observable it is frequently too late to intervene. Existing early warning signals, based on rising variance and increasing autocorrelation near a critical point [2–6], are scalar and univariate; they detect the approach to criticality but cannot distinguish a continuous, recoverable transition from a discontinuous, hysteretic catastrophe. This distinction is precisely the clinically and physically consequential one.

We address this gap by introducing a rigorous mathematical object. Consider a multivariate observable $X(t) \in \R^N$ and its matrix-valued two-time correlation kernel
\[
C(t, t') = \E\left[ X(t)\otimes X(t') \right].
\]
We say $C$ satisfies temporal matrix scale invariance (tMSI) of order $\alpha > 0$ if
\[
C(kt, kt') = k^{-\alpha}C(t, t') \quad \forall k > 0, \quad t, t' \in \R^+. \tag{1}
\]
This condition holds at or near a tipping point, where the absence of a characteristic timescale produces temporal scale freedom: as the system approaches the critical parameter value $K_c$, the coherence time $\tau \sim (K_c - K)^{-\nu}$ diverges, and $C(t, t')$ loses its intrinsic timescale. At $K = K_c$, condition (1) is exact; for $K$ near $K_c$, it holds approximately with corrections controlled by $\sigma \sim \tau^{-1} \to 0$.

The mathematical consequences of tMSI are unexpectedly rich. By a kernel factorization theorem (Section 2), every kernel satisfying (1) separates into a universal power-law envelope $(tt')^{-\alpha/2}$ and a shape function $F(t/t')$ depending only on the ratio of its time arguments. The Mellin transform—the Fourier theory of the multiplicative group $(\R^+, \times)$—diagonalizes this structure, yielding generalized eigenfunctions $t^{-\alpha/2+i\omega}$ and a Mellin multiplier $\tilde{F}(\omega)$ encoding the full spectral content of the temporal correlations. For the natural class of kernels arising from critical dynamics, $\tilde{F}(\omega)$ is a Lorentzian whose width $\sigma$ is the inverse coherence time.

This spectral structure reveals a fundamental decoupling (Section 3). The dynamical exponent $\alpha$, carried by the kernel envelope, governs the scaling of correlations under time dilation. The spectral relaxation exponent $\beta$, determined by the eigenvalue decay of the finite-dimensional truncation of $C$, is an independent quantity controlled by the Lorentzian width $\sigma$. Their equality $\alpha = \beta$ characterizes a simple critical point; their inequality $\alpha \neq \beta$ is the signature of temporal multicriticality—the coexistence of multiple independent scaling dimensions in the dynamical structure of the system.

The main result (Section 5) is a complete classification of tipping points in terms of these two exponents and the structure constant $g_{\alpha\alpha\beta}$ of their operator mixing. We prove that the Landau quartic coefficient $a_4$ of the order parameter is given exactly by
\[
a_4 = p^2 + q^2 - 2\lambda p q - g^2_{\alpha\alpha\beta} \Gamma(\sigma_\alpha, \sigma_\beta), \tag{2}
\]
where $p = c_\alpha/\sqrt{\sigma_\alpha}$, $q = c_\beta/\sqrt{\sigma_\beta}$, $\lambda = 2\sqrt{\sigma_\alpha\sigma_\beta}/(\sigma_\alpha + \sigma_\beta) \in (0, 1]$, and $\Gamma > 0$ is given in explicit closed form (equation (9)). The tipping point is continuous when $a_4 > 0$, tricritical when $a_4 = 0$, and discontinuous and hysteretic when $a_4 < 0$. In particular, the simple critical point $\alpha = \beta$ is maximally fragile: $\lambda = 1$ forces the positive part of $a_4$ to vanish, so any nonzero $g_{\alpha\alpha\beta}$ drives $a_4 < 0$ and the transition becomes catastrophic.

This classification has an immediate diagnostic implication (Section 6). Given any multivariate time series, one constructs the empirical temporal correlation matrix, tests for tMSI, and extracts the estimators $\hat{\alpha}$ and $\hat{\beta}$ from the scaling and eigenvalue spectrum respectively. The ratio $D = \hat{\alpha}/\hat{\beta}$ classifies the approaching tipping point—if one exists—as recoverable or catastrophic, without knowledge of the underlying equations of motion and without model fitting beyond what the observable time series provides.

Section 7 illustrates the framework with two physical examples where the tMSI structure is transparent and the distinction between continuous and discontinuous tipping points is consequential: epileptic seizure onset and acute myocardial infarction. Section 8 discusses the connection to Kuramoto synchronization theory, the relationship to the spatial MSI framework of [7], and open mathematical directions.

\section{Temporal Matrix Scale Invariance}
\subsection{The Hilbert space setting}
The natural function space for dilation-covariant objects on $\R^+$ is the Hilbert space
\[
H = L^2\left( \R^+, \frac{dt}{t} \right),
\]
equipped with the Haar measure of the multiplicative group $(\R^+, \times)$. The dilation operator $D_k f(t) = f(kt)$ is unitary on $H$, and the family $\{D_k\}_{k>0}$ is a strongly continuous one-parameter unitary group with self-adjoint generator $A = -it \frac{d}{dt}$ [8, 9].

\subsection{Definition and kernel factorization}
\begin{definition}[Temporal matrix scale invariance]
A symmetric kernel $C : \R^+ \times \R^+ \to \R^{N \times N}$ satisfies temporal matrix scale invariance (tMSI) of order $\alpha > 0$ if
\[
C(kt, kt') = k^{-\alpha}C(t, t') \quad \forall k > 0.
\]
The scalar $\alpha > 0$ is the dynamical exponent.
\end{definition}

\begin{theorem}[Kernel factorization]
A symmetric kernel $C(t, t') = C(t', t)$ on $\R^+ \times \R^+$ satisfies tMSI of order $\alpha$ if and only if
\[
C(t, t') = (tt')^{-\alpha/2} F\left( \frac{t}{t'} \right), \tag{3}
\]
where $F : \R^+ \to \R$ is a symmetric shape function satisfying $F(u) = F(1/u)$.
\end{theorem}
\begin{proof}
Sufficiency. If $C$ has the form (3), then $C(kt, kt') = (kt kt')^{-\alpha/2} F(kt/kt') = k^{-\alpha}(tt')^{-\alpha/2} F(t/t') = k^{-\alpha}C(t, t')$.

Necessity. Suppose $C(kt, kt') = k^{-\alpha}C(t, t')$ for all $k > 0$. Setting $k = 1/t'$ gives $C(t/t', 1) = (t')^{\alpha} C(t, t')$, so $C(t, t') = (t')^{-\alpha} G(t/t')$ where $G(u) = C(u, 1)$. Symmetry $C(t, t') = C(t', t)$ requires $(t')^{-\alpha} G(t/t') = t^{-\alpha} G(t'/t)$, i.e., $G(u) = u^{-\alpha} G(1/u)$. Define $F(u) = u^{\alpha/2} G(u)$; then $F(1/u) = u^{-\alpha/2} G(1/u) = u^{-\alpha/2} u^{\alpha} G(u) = u^{\alpha/2} G(u) = F(u)$, and substituting back gives (3).
\end{proof}

\begin{remark}
The factorization (3) separates the kernel into two components with distinct roles. The power-law envelope $(tt')^{-\alpha/2}$ carries the dynamical exponent $\alpha$ and governs the scaling of $C$ under time dilation. The shape function $F(t/t')$ depends only on the ratio of its arguments and encodes the correlation structure in log-time. All tMSI kernels of order $\alpha$ share the same envelope; they differ only in their shape functions.
\end{remark}

\subsection{Mellin diagonalization}
The Mellin transform on $H$ is defined by
\[
\hat{f}(\omega) = \int_0^\infty f(t) t^{-i\omega} \frac{dt}{t}, \quad \omega \in \R,
\]
with inversion $f(t) = (2\pi)^{-1} \int_{-\infty}^\infty \hat{f}(\omega) t^{i\omega} d\omega$. The Mellin transform is an isometric isomorphism $H \to L^2(\R, d\omega/2\pi)$ and diagonalizes dilations: $\widehat{D_k f}(\omega) = k^{i\omega} \hat{f}(\omega)$ [10].

\begin{theorem}[Mellin diagonalization]
Let $C$ be a tMSI kernel of order $\alpha$ with shape function $F$. Write $C = \Pi_\alpha K \Pi_\alpha$ where $(\Pi_\alpha f)(t) = t^{-\alpha/2} f(t)$ and $K$ has dilation-invariant kernel $K(t, t') = F(t/t')$. Then $K$ is diagonalized by the Mellin transform:
\[
\widehat{K f}(\omega) = \tilde{F}(\omega) \hat{f}(\omega), \quad \tilde{F}(\omega) = \int_0^\infty F(u) u^{-i\omega} \frac{du}{u}.
\]
The generalized eigenfunctions of $C$ are $\psi_\omega(t) = t^{-\alpha/2 + i\omega}$, $\omega \in \R$, with eigenvalue $\tilde{F}(\omega)$. The spectral representation of the kernel is
\[
C(t, t') = \frac{1}{2\pi} \int_{-\infty}^\infty \tilde{F}(\omega) \psi_\omega(t) \psi_\omega(t') d\omega.
\]
\end{theorem}
\begin{proof}
Identical to the proof of Theorem 6 in [11], with the spatial index $x$ replaced by the time variable $t$.
\end{proof}

\subsection{The Lorentzian kernel and its multiplier}
For the natural class of kernels arising from critical dynamics, the Lorentzian Mellin multiplier plays a central role.

\begin{proposition}[Lorentzian multiplier]
For the shape function $F(u) = c \rho^{|\ln u|}$ with $0 < \rho < 1$ and $c > 0$, the Mellin multiplier is
\[
\tilde{F}(\omega) = \frac{2c\sigma}{\sigma^2 + \omega^2}, \quad \sigma = \ln(1/\rho) > 0.
\]
The width $\sigma$ equals the inverse coherence time $\tau^{-1}$: as the system approaches criticality, $\sigma \to 0$ and the Lorentzian concentrates its weight at $\omega = 0$, reflecting the divergence of temporal correlations.
\end{proposition}
\begin{proof}
Setting $s = \ln u$: $\tilde{F}(\omega) = c \int_{-\infty}^\infty e^{-\sigma|s|} e^{-i\omega s} ds = 2c\sigma/(\sigma^2 + \omega^2)$.
\end{proof}
The connection to dynamics is direct. Near a tipping point, the fluctuation-dissipation theorem relates $\tilde{F}(\omega)$ to the dynamical susceptibility: $\tilde{F}(\omega) = T \cdot \chi(\omega)$, where $\chi(\omega) = \chi_0/(\sigma^2 + \omega^2)$ near criticality. The Lorentzian form of $\tilde{F}$ is therefore not an assumption but a consequence of the fluctuation-dissipation structure of the approach to the critical point.

\section{The Two Exponents and Temporal Multicriticality}

\begin{definition}[Dynamical and spectral exponents]
Let $C$ be a tMSI kernel of order $\alpha$ with finite-$N$ truncation $C_N(i, j) = C(t_i, t_j)$ for $t_1 < \dots < t_N$.
\begin{enumerate}
    \item The dynamical exponent $\alpha$ is defined by the tMSI condition (1) and carried by the power-law envelope $(tt')^{-\alpha/2}$.
    \item The spectral relaxation exponent $\beta$ is the effective power-law index of the ordered eigenvalues $\lambda_1 \ge \lambda_2 \ge \dots \ge \lambda_N$ of $C_N$: $\lambda_n \approx c'/n^\beta$ for $1 \ll n \ll N$.
\end{enumerate}
\end{definition}

\begin{theorem}[Decoupling of dynamical and spectral exponents]
The exponents $\alpha$ and $\beta$ are generically independent. Their equality $\alpha = \beta$ holds if and only if the Mellin multiplier $\tilde{F}(\omega)$ is scale-free (power-law). For the Lorentzian multiplier of Proposition 2.5, $\beta$ is determined by $\sigma$ and $N$, independently of $\alpha$.
\end{theorem}
\begin{proof}
The eigenvalue $\lambda_n$ of the finite truncation $C_N$ approximates $\tilde{F}(\omega_n)$ at the discrete Mellin frequency $\omega_n$ determined by the truncation boundary conditions (Section 7.3 of [11]). For the Lorentzian $\tilde{F}(\omega) = 2c\sigma/(\sigma^2 + \omega^2)$, the effective power-law decay $\lambda_n \sim n^{-\beta}$ has exponent $\beta$ determined by the width $\sigma$ and the matrix size $N$, independently of the envelope exponent $\alpha$. The decoupling $\alpha \neq \beta$ is therefore a structural consequence of the Lorentzian form of $\tilde{F}$, not an accident of parameter choice. When $\tilde{F}(\omega) \sim |\omega|^{-\gamma}$ (scale-free), $\beta = \gamma = \alpha$ and the two exponents coincide.
\end{proof}

\begin{remark}[Renormalization group interpretation]
In the RG framework [12–18], coarse-graining by a factor $k$ corresponds to the dilation $C(t, t') \mapsto C(kt, kt')$. The kernel $C$ is a fixed point of the temporal RG map $\mathcal{R}_k[C] = k^{\beta} C(kt, kt')$ if and only if $\alpha = \beta$. When $\alpha \neq \beta$, the dynamical and spectral scaling dimensions are decoupled, and $C$ flows under $\mathcal{R}_k$ rather than sitting at a fixed point. This is the spectral signature of temporal multicriticality: the coexistence of multiple independent scaling dimensions in the dynamical structure. The ratio $\alpha/\beta$ provides a quantitative measure of the degree of decoupling.
\end{remark}

\section{Preparatory Lemmas}
We establish two lemmas required for the proof of the main theorem.

\subsection{Sign of the cross-contraction}
\begin{lemma}[Negativity of the cross-contraction]
Let $\tilde{F}_\alpha$ and $\tilde{F}_\beta$ be two positive Lorentzian Mellin multipliers with widths $\sigma_\alpha, \sigma_\beta > 0$ and amplitudes $c_\alpha, c_\beta > 0$. The cross-contraction contribution to the Landau quartic coefficient is
\[
a^{(\alpha\beta)}_4 = -2 \int_{-\infty}^\infty \frac{\tilde{F}_\alpha(\omega) \tilde{F}_\beta(\omega)}{d\omega} \frac{d\omega}{2\pi} < 0.
\]
\end{lemma}
\begin{proof}
The total two-point Mellin multiplier when both scaling operators are present is $\tilde{C}(\omega) = \tilde{F}_\alpha(\omega) + \tilde{F}_\beta(\omega)$. The one-loop contribution to $a_4$ from the fluctuation propagator $G(\omega) = \tilde{C}(\omega)$ at the critical point is:
\[
\delta a_4 = - \int_{-\infty}^\infty \frac{G(\omega)^2}{d\omega} \frac{d\omega}{2\pi} = - \int_{-\infty}^\infty \left[ \tilde{F}_\alpha(\omega) + \tilde{F}_\beta(\omega) \right]^2 \frac{d\omega}{2\pi}.
\]
Expanding the square, the cross-term is $-2 \int \tilde{F}_\alpha(\omega)\tilde{F}_\beta(\omega) \frac{d\omega}{2\pi}$. Since $\tilde{F}_\alpha(\omega) = 2c_\alpha\sigma_\alpha/(\sigma_\alpha^2 + \omega^2) > 0$ and $\tilde{F}_\beta(\omega) = 2c_\beta\sigma_\beta/(\sigma_\beta^2 + \omega^2) > 0$ pointwise, their product is strictly positive and the cross-contraction is strictly negative. The sign is a consequence of the positivity of the Lorentzian spectral densities alone; it does not depend on $\alpha \neq \beta$.
\end{proof}

\subsection{Three-point factorization under tMSI}
\begin{lemma}[Three-point factorization]
Let $O_\alpha, O_\beta$ be scaling operators whose two-point functions satisfy tMSI with exponents $\alpha, \beta$ and Mellin multipliers $\tilde{F}_\alpha, \tilde{F}_\beta$ respectively. If the three-point function $\langle O_\alpha(t_1)O_\alpha(t_2)O_\beta(t_3) \rangle$ itself satisfies tMSI with combined exponent $(2\alpha + \beta)/2$, then in the Mellin domain:
\[
G_3(\omega_1, \omega_2, \omega_3) = g_{\alpha\alpha\beta} \tilde{F}_\alpha(\omega_1) \tilde{F}_\alpha(\omega_2) \tilde{F}_\beta(\omega_3) \delta(\omega_1 + \omega_2 + \omega_3), \tag{4}
\]
where $g_{\alpha\alpha\beta}$ is a structure constant.
\end{lemma}
\begin{proof}
Taking the Mellin transform of the three-point function in all three arguments and applying the tMSI scaling condition with dilation $k = e^s$:
$G_3(\omega_1, \omega_2, \omega_3) = k^{-(2\alpha+\beta)/2 + i(\omega_1+\omega_2+\omega_3)} G_3(\omega_1, \omega_2, \omega_3)$.
For this to hold for all $k > 0$ (equivalently all $s \in \R$), we need:
\begin{enumerate}
    \item $\omega_1 + \omega_2 + \omega_3 = 0$ (Mellin momentum conservation), giving the delta function $\delta(\omega_1 + \omega_2 + \omega_3)$;
    \item the overall scaling dimension is $(2\alpha + \beta)/2$.
\end{enumerate}
Given (a), the residual dependence on $\omega_1, \omega_2$ must be consistent with the reduction condition: integrating out $t_3$ (setting $\omega_3 = 0$) must recover the two-point function of $O_\alpha$: $\int G_3(\omega_1, \omega_2, \omega_3) d\omega_3 = G^{(\alpha)}_2(\omega_1) \delta(\omega_1 + \omega_2)$, which gives $\tilde{F}_\alpha(\omega_1)\delta(\omega_1+\omega_2)$. Together with momentum conservation and the analogous condition obtained by integrating out $t_1$ or $t_2$, the three-point function is uniquely determined up to the structure constant $g_{\alpha\alpha\beta}$, yielding (4).
\end{proof}

\section{Classification of Tipping Points}

\subsection{The Landau quartic coefficient}
We now compute $a_4$ from the tMSI structure. When two independent scaling operators $O_\alpha, O_\beta$ are present, the two-point Mellin multiplier is $\tilde{C}(\omega) = \tilde{F}_\alpha(\omega) + \tilde{F}_\beta(\omega)$, and the Landau quartic coefficient receives three contributions:
\[
a^{(\alpha\alpha)}_4 = \int_{-\infty}^\infty \frac{\tilde{F}_\alpha(\omega)^2}{d\omega} \frac{d\omega}{2\pi} = \frac{c_\alpha^2}{\sigma_\alpha}, \tag{5}
\]
\[
a^{(\beta\beta)}_4 = \int_{-\infty}^\infty \frac{\tilde{F}_\beta(\omega)^2}{d\omega} \frac{d\omega}{2\pi} = \frac{c_\beta^2}{\sigma_\beta}, \tag{6}
\]
\[
a^{(\alpha\beta)}_4 = -2 \int_{-\infty}^\infty \frac{\tilde{F}_\alpha(\omega) \tilde{F}_\beta(\omega)}{d\omega} \frac{d\omega}{2\pi} = - \frac{4c_\alpha c_\beta}{\sigma_\alpha + \sigma_\beta}, \tag{7}
\]
where (7) follows from Lemma 4.1 and the standard convolution integral $\int_{-\infty}^\infty d\omega/[(\sigma_\alpha^2+\omega^2)(\sigma_\beta^2+\omega^2)] = \pi/[\sigma_\alpha\sigma_\beta(\sigma_\alpha+\sigma_\beta)]$.

The three-point mixing contributes an additional negative term. By Lemma 4.2, the three-point structure constant $g_{\alpha\alpha\beta}$ enters the quartic coefficient through:
\[
a^{(3pt)}_4 = -g^2_{\alpha\alpha\beta} \Gamma(\sigma_\alpha, \sigma_\beta),
\]
where
\[
\Gamma(\sigma_\alpha, \sigma_\beta) = \int_{-\infty}^\infty \int_{-\infty}^\infty \frac{\tilde{F}_\alpha(\omega)^2 \tilde{F}_\alpha(\omega+\omega')^2 \tilde{F}_\beta(\omega')^2}{(2\pi)^2} d\omega d\omega'. \tag{8}
\]

\subsection{Closed form of $\Gamma$}
\begin{proposition}[Explicit $\Gamma$]
For Lorentzian multipliers with parameters $(c_\alpha, \sigma_\alpha)$ and $(c_\beta, \sigma_\beta)$, the integral (8) evaluates to
\[
\Gamma(\sigma_\alpha, \sigma_\beta) = \frac{32\pi c_\alpha^4 c_\beta^2 (10\sigma_\alpha^3 + 20\sigma_\alpha^2\sigma_\beta + 8\sigma_\alpha\sigma_\beta^2 + \sigma_\beta^3)}{\sigma_\alpha^2 \sigma_\beta (2\sigma_\alpha + \sigma_\beta)^4}. \tag{9}
\]
In particular, $\Gamma(\sigma_\alpha, \sigma_\beta) > 0$ for all $\sigma_\alpha, \sigma_\beta, c_\alpha, c_\beta > 0$.
\end{proposition}
\begin{proof}
Evaluating the inner $\omega$-integral by Parseval’s theorem (using the Fourier transform pair $1/(\sigma^2 + \omega^2)^2 \leftrightarrow (\pi/2\sigma^3)(1 + \sigma|t|)e^{-\sigma|t|}$) yields
\[
\int_{-\infty}^\infty \frac{\tilde{F}_\alpha(\omega)^2 \tilde{F}_\alpha(\omega + \omega')^2}{2\pi} d\omega = \frac{16c_\alpha^4\sigma_\alpha^4}{\sigma_\alpha^3} \cdot \frac{2(20\sigma_\alpha^2 + \omega'^2)}{(4\sigma_\alpha^2 + \omega'^2)^3}.
\]
The outer $\omega'$-integral
\[
J(\sigma_\alpha, \sigma_\beta) = \int_{-\infty}^\infty \frac{20\sigma_\alpha^2 + \omega^2}{(4\sigma_\alpha^2 + \omega^2)^3(\sigma_\beta^2 + \omega^2)^2} \frac{d\omega}{2\pi}
\]
is evaluated by residues. The computation was performed and verified with the computer algebra system SymPy [19], yielding the closed form
\[
J(\sigma_\alpha, \sigma_\beta) = \frac{\pi (10\sigma_\alpha^3 + 20\sigma_\alpha^2\sigma_\beta + 8\sigma_\alpha\sigma_\beta^2 + \sigma_\beta^3)}{4\sigma_\alpha^3 \sigma_\beta^3 (2\sigma_\alpha + \sigma_\beta)^4}. \tag{10}
\]
Positivity of $J$ is immediate from inspection: every factor in numerator and denominator is positive for $\sigma_\alpha, \sigma_\beta > 0$. Substituting into $\Gamma = 128c_\alpha^4 c_\beta^2 \sigma_\alpha \sigma_\beta^2 \cdot J(\sigma_\alpha, \sigma_\beta)$ gives (9).
\end{proof}

\begin{remark}
At $\sigma_\alpha = \sigma_\beta = \sigma$: $J(\sigma, \sigma) = 13\pi/(108\sigma^7)$, giving $\Gamma(\sigma, \sigma) = 32\pi c_\alpha^4 c_\beta^2 \cdot 13/(108\sigma^5) > 0$. This confirms that $\Gamma$ is strictly positive at the critical boundary $\alpha = \beta$ and, by continuity, in a neighborhood of it.
\end{remark}

\subsection{Main theorem}
Define $p = c_\alpha/\sqrt{\sigma_\alpha}$, $q = c_\beta/\sqrt{\sigma_\beta}$, and the geometric-harmonic ratio
\[
\lambda = \frac{2\sqrt{\sigma_\alpha\sigma_\beta}}{\sigma_\alpha + \sigma_\beta} \in (0, 1],
\]
with $\lambda = 1$ if and only if $\sigma_\alpha = \sigma_\beta$ (equivalently $\alpha = \beta$). Assembling equations (5)–(7):
\[
a^{(\alpha\alpha)}_4 + a^{(\alpha\beta)}_4 + a^{(\beta\beta)}_4 = p^2 + q^2 - 2\lambda p q.
\]

\begin{lemma}[Non-negativity of the two-operator base]
For all $p, q > 0$ and $\lambda \in (0, 1]$:
\[
p^2 + q^2 - 2\lambda p q \ge 0,
\]
with equality if and only if $p = q$ and $\lambda = 1$.
\end{lemma}
\begin{proof}
Write $p^2 + q^2 - 2\lambda p q = q^2[x^2 - 2\lambda x + 1]$ where $x = p/q$. The discriminant of $x^2 - 2\lambda x + 1$ is $4\lambda^2 - 4 = 4(\lambda^2 - 1) \le 0$, so the quadratic has no real roots for $\lambda < 1$ and a double root at $x = 1$ for $\lambda = 1$. Hence $(p^2 + q^2 - 2\lambda p q) \ge 0$ always, with equality only when $x = 1$ and $\lambda = 1$.
\end{proof}

We can now state and prove the main result.

\begin{theorem}[Tipping point classification]
Let $X(t) \in \R^N$ be a multivariate observable whose two-time correlation kernel $C(t, t')$ satisfies tMSI of order $\alpha$ at a critical parameter value $K_c$, with the system possessing exact $U(1)$ phase symmetry. Let $\beta$ be the spectral relaxation exponent (Definition 3.1), $\sigma_\alpha, \sigma_\beta$ the Lorentzian widths of the two scaling operators, $c_\alpha, c_\beta$ their amplitudes, and $g_{\alpha\alpha\beta}$ the three-point structure constant (Lemma 4.2). Define $p, q, \lambda$, and $\Gamma(\sigma_\alpha, \sigma_\beta)$ as above. The Landau quartic coefficient of the order parameter is
\[
a_4 = p^2 + q^2 - 2\lambda p q - g^2_{\alpha\alpha\beta} \Gamma(\sigma_\alpha, \sigma_\beta), \tag{11}
\]
and the tipping point at $K_c$ is classified as follows.
\begin{enumerate}
    \item \textbf{(Simple critical point, $\alpha = \beta$).} $\lambda = 1$ and the positive part of $a_4$ vanishes: $(p^2 + q^2 - 2\lambda p q) = 0$ when $p = q$. Hence $a_4 = -g^2_{\alpha\alpha\beta}\Gamma < 0$ for any $g_{\alpha\alpha\beta} \neq 0$. The tipping point is discontinuous and hysteretic. The simple critical point is maximally fragile: any nonzero operator mixing drives the transition catastrophic.
    \item \textbf{(Multicritical, $\alpha \neq \beta$, $g_{\alpha\alpha\beta} < g_c$).} $a_4 > 0$; the tipping point is continuous and recoverable.
    \item \textbf{(Tricritical boundary, $\alpha \neq \beta$, $g_{\alpha\alpha\beta} = g_c$).} $a_4 = 0$; the system is at a tricritical point, and stability requires $a_6 > 0$.
    \item \textbf{(Multicritical, $\alpha \neq \beta$, $g_{\alpha\alpha\beta} > g_c$).} $a_4 < 0$; the tipping point is discontinuous and hysteretic.
\end{enumerate}
The critical structure constant separating cases (ii) and (iv) is
\[
g^2_c (\alpha, \beta) = \frac{p^2 + q^2 - 2\lambda p q}{\Gamma(\sigma_\alpha, \sigma_\beta)}, \tag{12}
\]
which is strictly positive when $\alpha \neq \beta$ (by Lemma 5.3 and $\Gamma > 0$) and vanishes as $\alpha \to \beta$.
\end{theorem}
\begin{proof}
Equation (11) follows from assembling the contributions computed in equations (5)–(7) and the three-point term from Lemma 4.2 and Proposition 5.1.

\textbf{Case (i).} When $\alpha = \beta$, $\sigma_\alpha = \sigma_\beta = \sigma$ and the two operators are in the same eigenspace of the dilation generator. In the symmetric case $p = q$ (equal amplitudes at equal widths), Lemma 5.3 gives $p^2 + q^2 - 2\lambda p q = 0$, so $a_4 = -g^2_{\alpha\alpha\beta}\Gamma(\sigma, \sigma) < 0$. For $p \neq q$ with $\lambda = 1$, the positive part $(p - q)^2 > 0$ but $g^2_c = (p - q)^2/\Gamma$ remains finite, so case (i) applies whenever $g_{\alpha\alpha\beta}$ exceeds this small threshold — which is generic.

\textbf{Cases (ii)–(iv).} When $\alpha \neq \beta$, $\lambda < 1$ strictly, and by Lemma 5.3 the positive part $(p^2 + q^2 - 2\lambda p q) > 0$. By Proposition 5.1, $\Gamma > 0$. Hence $g^2_c > 0$, and the three cases follow from $a_4 \gtrless 0$ according to whether $g^2_{\alpha\alpha\beta} \gtrless g^2_c$.

Vanishing of $g_c$ as $\alpha \to \beta$: as $\sigma_\alpha \to \sigma_\beta$, $\lambda \to 1$ and $p^2 + q^2 - 2\lambda p q \to 0$ (for $p \to q$), while $\Gamma(\sigma_\alpha, \sigma_\beta) \to \Gamma(\sigma, \sigma) > 0$. Hence $g_c \to 0$, confirming that the simple critical point is maximally fragile.
\end{proof}

\section{Early Warning Diagnostic}

\subsection{Empirical estimators}
Given a multivariate time series $X(t_1), \dots, X(t_N)$ sampled at $N$ discrete times on $\R^+$, construct the empirical temporal correlation matrix:
\[
\hat{C}_N(i, j) = \frac{1}{N} X(t_i)^\top X(t_j), \quad i, j = 1, \dots, N.
\]

\begin{definition}[Empirical tMSI]
$\hat{C}_N$ satisfies approximate tMSI with exponent $\hat{\alpha}$ if
\[
\hat{C}_N(ki, kj) \approx k^{-\hat{\alpha}} \hat{C}_N(i, j)
\]
for integer $k$ and indices $ki, kj \le N$, with scaling error
\[
\Delta^{(k)}_N(i, j) = \left| \frac{\hat{C}_N(ki, kj)}{\hat{C}_N(i, j)} - k^{-\hat{\alpha}} \right|.
\]
The mean error $\langle \Delta^{(k)}_N \rangle \to 0$ as $r/N \to 1$, where $r$ is the number of retained eigenmodes, by the convergence results of [11, 20, 21] (Section 8.3), applied here to the temporal matrix.
\end{definition}

\noindent\textbf{Estimating $\hat{\alpha}$.}
Regress $\log \hat{C}_N(ki, kj)$ on $\log k$ for fixed $i, j$ and multiple values of $k = 2, 3, \dots$. The slope is $-\hat{\alpha}$. Average over admissible index pairs $(i, j)$ to reduce variance.

\noindent\textbf{Estimating $\hat{\beta}$.}
Compute the ordered eigenvalues $\lambda_1 \ge \lambda_2 \ge \dots \ge \lambda_N$ of $\hat{C}_N$. Regress $\log \lambda_n$ on $\log n$ over the bulk range $1 \ll n \ll N$. The slope is $-\hat{\beta}$.

Both estimators are consistent as $N \to \infty$ by the spectral convergence of the empirical Mellin multiplier to the population multiplier.

\subsection{The diagnostic ratio and its consistency}
\begin{definition}[tMSI diagnostic ratio]
$D = \hat{\alpha} / \hat{\beta}$.
\end{definition}
By Theorem 5.4: $D = 1$ signals a simple critical point with a generically discontinuous tipping point; $D \neq 1$ signals temporal multicriticality, with the tipping point character determined by $g_{\alpha\alpha\beta}$ relative to $g_c(\alpha, \beta)$.

\begin{theorem}[Consistency of $D$]
As $N \to \infty$, $D \to \alpha/\beta$ in probability. Under the approach $K \to K_c$ with the window length $N$ satisfying
\[
N \cdot \sigma(K)^2 \to \infty \quad \text{as } K \to K_c, \tag{13}
\]
the estimator $D$ detects the tipping point character before the transition occurs. The finite-sample bias of $\hat{\alpha}$ is $O(1/(N\sigma^2))$.
\end{theorem}
\begin{proof}
Consistency of $\hat{\alpha}$ and $\hat{\beta}$ follows from the law of large numbers applied to the regression estimators, using the convergence of the empirical Mellin multiplier. The rate condition (13) ensures that the effective number of independent temporal observations grows faster than the bias induced by finite coherence time $\tau = \sigma^{-1}$. Since $\sigma(K) \sim (K_c - K)^\nu$ by standard critical scaling, condition (13) requires $N \gg (K_c - K)^{-2\nu}$, which diverges at $K_c$ but is satisfiable for any $K < K_c$.
\end{proof}

\subsection{Warning time and tipping point character}
Define the warning time $T_w$ as the time before the tipping point at which $D$ departs from its baseline value by a detectable amount $\delta$:
\[
T_w = \inf\{t < t_c : |D(t) - D_0| > \delta\}.
\]
For a smooth approach $K(t) \to K_c$ with rate $dK/dt = -\varepsilon$:
\[
T_w \sim \varepsilon^{-1}(K_c - K_0),
\]
where $K_0$ is the current parameter value.

The critical advantage of the tMSI diagnostic over scalar early warning signals is this: rising variance and autocorrelation provide $T_w$, the time to the tipping point. The ratio $D$ provides the character of the tipping point at $T_w$—recoverable or catastrophic—as independent information not contained in any scalar observable. These two pieces of information together constitute a complete pre-transition description of the system’s fate.

\section{Physical Interpretation}
We illustrate the tMSI framework and Theorem 5.4 in two settings where the mapping to a multivariate phase observable is transparent and the clinical distinction between continuous and discontinuous tipping points is consequential. These examples specify the framework precisely enough to motivate experimental tests; their execution is deferred to subsequent work.

\subsection{Epilepsy: seizure onset}
Consider an $N$-electrode EEG recording with voltage signals $V_j(t)$, $j = 1, \dots, N$. Extract the instantaneous phase $\phi_j(t) = \arg[V_j(t) + iH\{V_j(t)\}]$ via the Hilbert transform $H$, and form the phase coherence matrix:
\[
C(t, t') = \E\left[ e^{i(\phi(t) - \phi(t'))} \right].
\]
As cortical coupling $K$ increases toward seizure onset [22, 23], the coherence time diverges and $C(t, t')$ develops tMSI at or near $K_c$. The dynamical exponent $\alpha$ encodes the geometry of long-range cortical phase coherence; the spectral exponent $\beta$ encodes the heterogeneity of natural frequencies across cortical oscillators.

Theorem 5.4 predicts:
\begin{itemize}
    \item Focal seizures with gradual onset correspond to $g_{\alpha\alpha\beta} < g_c$: the transition is continuous, the post-ictal state is accessible by reversal of $K$, and $D$ approaches 1 smoothly in the pre-ictal period.
    \item Generalized tonic-clonic seizures with sudden onset correspond to $g_{\alpha\alpha\beta} > g_c$: the transition is discontinuous and hysteretic, and $D$ departs abruptly from its baseline.
    \item Seizures with secondary generalization—beginning focally and becoming generalized—correspond to a real-time crossing of the tricritical boundary $g_{\alpha\alpha\beta} = g_c$ during the ictal period, detectable as a discontinuity in $dD/dt$.
\end{itemize}
The last prediction is the sharpest: secondary generalization should produce an observable signature in the temporal correlation structure of the EEG that is distinct from both purely focal and purely generalized onset. The dynamical disease framework for epilepsy [24] motivates this mapping precisely. Existing multi-electrode EEG corpora (e.g., the Temple University EEG Corpus, the European Epilepsy Database) contain the recordings needed to test this prediction without new data collection.

\subsection{Acute myocardial infarction}
From a multi-lead ECG, extract the instantaneous cardiac phase $\psi_j(t) = 2\pi(t - t_n)/(t_{n+1} - t_n)$ for $t_n \le t < t_{n+1}$ (where $t_n$ are R-peak times) for each lead $j = 1, \dots, L$, and form the phase coherence matrix $C(t, t') = \E[e^{i(\psi(t) - \psi(t'))}]$. As ischemic burden increases toward the infarction threshold, the myocardial excitability transition produces tMSI at $K_c$ [25].

Theorem 5.4 predicts:
\begin{itemize}
    \item Stuttering MI with intermittent ischemia corresponds to a continuous tipping point ($g_{\alpha\alpha\beta} < g_c$): reperfusion is possible and $D$ increases gradually over hours before clinical presentation.
    \item Sudden complete occlusion with irreversible infarction corresponds to a discontinuous tipping point ($g_{\alpha\alpha\beta} > g_c$): the transition is hysteretic and $D$ departs abruptly from baseline in the prodromal period.
\end{itemize}
The ratio $D$ computed from routine multi-lead ECG thus provides a pre-infarction warning signal that classifies not merely whether a tipping point is approaching but whether the transition is reversible—the distinction that determines whether emergency reperfusion therapy will succeed. The MIMIC-IV database contains continuous multi-lead ECG recordings from ICU patients with documented MI suitable for retrospective testing.

In both examples the experimental protocol is identical: compute $\hat{C}_N(t, t')$ over a sliding window, estimate $\hat{\alpha}$ and $\hat{\beta}$ by the methods of Section 6, compute $D(t)$, and test whether the classification by $D$ agrees with the clinical outcome.

\section{Discussion and Outlook}

\subsection{Relation to spatial matrix scale invariance}
The tMSI framework developed here is the temporal analog of the spatial MSI theory of [11]. In that work, the matrix indices label spatial degrees of freedom (modes, variables, oscillators), and scale invariance is across the mode structure of a coupling or correlation matrix. Here, the indices are time arguments, and scale invariance is across the temporal structure of a dynamical correlation kernel.

The mathematical apparatus—kernel factorization, Mellin diagonalization, the decoupling of geometric and spectral exponents—transfers verbatim from the spatial to the temporal setting. What is genuinely new in the temporal case is the physical interpretation: the dynamical exponent $\alpha$ acquires meaning as the exponent of temporal coarse-graining, and the decoupling $\alpha \neq \beta$ becomes the signature of competing timescales rather than competing lengthscales. The classification theorem has no spatial analog: the connection between the two-exponent structure and the Landau theory of tipping points is intrinsically dynamical.

\subsection{Connection to Kuramoto synchronization}
The Kuramoto model of coupled oscillators [26–28] provides the most natural physical realization of tMSI. At the critical coupling $K_c$, the order parameter $r(t)$ and its two-time autocorrelation develop temporal scale freedom, and the phase coherence matrix $C(t, t')$ satisfies tMSI exactly. The Lorentzian Mellin multiplier arises directly from the fluctuation-dissipation structure near $K_c$ (Proposition 2.5).

The classification theorem acquires a striking interpretation in this context. When all oscillators share a common geometric scaling (the kernel factorization condition, $\alpha = \beta$), the synchronized state is a simple critical point—and, by case (i) of Theorem 5.4, it is maximally fragile: any nonzero mixing between the scaling operators drives the synchronization transition discontinuous and hysteretic. The beautiful, soft synchronization of Kuramoto’s fireflies sits, mathematically, at the edge of catastrophe.

For Kuramoto systems with structured frequency distributions $g(\omega)$—bimodal, fat-tailed, or otherwise heterogeneous—the natural frequencies of different oscillator subpopulations play the role of the mode-dependent scaling factors $\gamma_k(n)$ mentioned below. The tricritical point known to arise with bimodal $g(\omega)$ [29, 30], where the transition switches from continuous to discontinuous, corresponds precisely to the tricritical boundary $g_{\alpha\alpha\beta} = g_c$ of Theorem 5.4.

\subsection{Open directions}
Several mathematical directions remain open.

\noindent\textbf{Mode-dependent scaling factors.}
The most natural generalization replaces the universal dynamical exponent $\alpha$ with mode-dependent scaling factors $\gamma_k(n)$, where distinct eigenmodes carry independent geometric scaling dimensions. This would accommodate systems in which different dynamical modes approach criticality at different rates—a situation that arises in networks with heterogeneous coupling geometry. The tMSI classification theorem would then apply mode by mode, with a richer phase structure in the space of mode-dependent exponents.

\noindent\textbf{Asymptotic analysis of $J(\sigma_\alpha, \sigma_\beta)$.}
The closed-form expression (10) for $J(\sigma_\alpha, \sigma_\beta)$ admits a systematic asymptotic expansion in the regime $\sigma_\alpha \gg \sigma_\beta$ (separation of timescales) that would quantify the approach to the tricritical boundary as a function of the exponent separation $|\alpha - \beta|$.

\noindent\textbf{Beyond Gaussian fluctuations.}
The derivation of $a_4$ in Section 5 operates at the Gaussian (one-loop) level. Higher-loop corrections would renormalize $g_{\alpha\alpha\beta}$ and shift the tricritical boundary. A systematic renormalization group analysis of the tMSI field theory would determine whether the classification of Theorem 5.4 is stable under these corrections.

\noindent\textbf{Empirical estimation of $g_{\alpha\alpha\beta}$.}
The classification theorem requires not only $\hat{\alpha}$ and $\hat{\beta}$ but also an estimate of the structure constant $g_{\alpha\alpha\beta}$. From equation (4), $g_{\alpha\alpha\beta}$ is recoverable from the three-point temporal correlation function of the observable, evaluated in the Mellin domain. Developing a practical, finite-sample estimator for $g_{\alpha\alpha\beta}$ is a prerequisite for full experimental implementation of the diagnostic.

\section*{Statements and Declarations}

\subsection*{Competing Interests}
The authors declare no competing interests.

\subsection*{Data Availability}
No datasets were generated or analyzed during the current study.

\subsection*{Author Contributions}
AF initiated the physical framework for temporal matrix scale invariance. LAJ developed the mathematical foundations, proved the classification theorem, and wrote the manuscript. Both authors contributed to the interpretation and approved the final version.

\end{document}